\documentclass{llncs}

\usepackage{latexsym}
\usepackage{epsfig}
\usepackage{amssymb}
\usepackage{array}
\usepackage{float}

\graphicspath{{figures/}{./}}

\newtheorem{observation}{Observation}
\floatstyle{ruled}
\newfloat{algo}{thp}{lop}%[section]
\floatname{algo}{Algorithm}
\newcommand{\IF}[1]{\textbf{if} {#1}}
\newcommand{\THEN}{\textbf{then} }
\newcommand{\ELSE}{\textbf{else} }

\newcommand{\ENDIF}{\textbf{endif} }

\newcounter{ALGOline}
\newcounter{ALGOAlgo}
\newcommand{\LINESEP}{.}
\newcommand{\LINESTYLE}{\tiny}
\newcommand{\INITALGO}[1]{\setcounter{ALGOAlgo}{#1}}
\newcommand{\INITLINE}{\setcounter{ALGOline}{1}}
\newcommand{\NA}{{\addtocounter{ALGOAlgo}{1}{\INITLINE}}}
\newcommand{\AL}{
    \LINESTYLE{\arabic{ALGOAlgo}\LINESEP
    \ifnum \value{ALGOline}<10 0\fi
    \LINESTYLE{\arabic{ALGOline}}
    \addtocounter{ALGOline}{1}
}}

\newcommand{\BEGLIST}{\begin{list}{}{\partopsep -3pt \parsep -2pt \listparindent -0pt \labelwidth .5in}}
\newcommand{\ENDLIST}{\end{list}}

\newcommand{\ie}{\emph{i.e., }}
\newcommand{\eg}{\emph{e.g., }}

\begin{document}

\INITALGO{0}
\INITLINE

%\date{\today}
\title{Self-stabilizing Deterministic Gathering} 

\titlerunning{Self-stabilizing Deterministic Gathering}

\author{Yoann Dieudonn\'e\inst{1} \qquad Franck Petit\inst{2}}

\authorrunning{Dieudonn\'e and Petit}   % abbreviated author list (for running head)

\tocauthor{Yoann Dieudonn\'e, Franck Petit}

\institute{MIS CNRS,  Universit\'{e} de Picardie Jules Verne Amiens, France
\and
INRIA, LIP UMR 5668, Universit\'e de Lyon / ENS Lyon, France}
\date{}
\maketitle

\begin{abstract}
In this paper, we investigate the possibility to deterministically solve the gathering problem (GP)
with weak robots (anonymous, autonomous, disoriented, deaf and dumb, and oblivious).  
We introduce strong multiplicity detection as the ability for the robots to detect 
the exact number of robots located at a given position. 
We show that with strong multiplicity detection, there exists a deterministic 
self-stabilizing algorithm solving GP for $n$ robots if, and only if, $n$ is odd.

\end{abstract}
	
\section{Introduction}

The distributed systems considered in this paper are teams (or swarms) of {\em mobile} robots (sensors or agents).  
Such systems supply the ability to collect (to sense) environmental data such as 
temperature, sound, vibration, pressure, motion, etc.
The robots use these sensory data as an input in order to act in a given (sometimes dangerous) physical environment.
Numerous potential applications exist for such multi-robot systems, \eg environmental monitoring, large-scale construction,
risky area surrounding, exploration of an unknown area. 
All these applications involve basic cooperative tasks such as pattern formation, 
gathering, scatter, leader election, flocking, etc.

Among the above fundamental coordination tasks, we address the {\em gathering} 
(or {\em Rendez-Vous}) problem.  This problem can be stated as follows: robots, initially located 
at various positions, gather at the same position in finite time and remain at this position thereafter. 
%consisting to gather a team of mobile robots located 
%at various positions at a single position in finite time.
The difficulty to solve this problem greatly depends on the system settings, \eg 
whether the robots can remember past events or not, their means of communication, their ability
to share a global property like observable IDs, sense of direction, global coordinate, etc.  For
instance, assuming that the robots share a common global coordinate system or have (observable) IDs
allowing to differentiate any of them, 
it is easy to come up with a deterministic distributed algorithm for that problem.  
Gathering turns out to be very difficult to solve with {\em weak} robots, \ie devoid of 
$(1)$ any (observable) IDs allowing to differentiate any of them ({\em anonymous}), 
$(2)$ any central coordination mechanism or scheduler ({\em autonomous}),
$(3)$ any common coordinate mechanism or common sense of direction ({\em disoriented}), 
$(4)$ means of communication allowing them to communicate directly, \eg by radio frequency ({\em deaf and dumb}), and  
$(5)$ any way to remember any previous observation nor computation performed in any previous step ({\em oblivious}).
Every movement made by a robot is then the result of a computation having 
observed positions of the other robots as a only possible input. 
With such settings, assuming that robots are points evolving on the plane, 
no solution exists for the gathering problem if the system contains two robots only~\cite{SY99}.  
It is also shown in~\cite{Prencipe07} that gathering can be
solved only if the robots have the capability to know whether several robots are located at 
the same position ({\em multiplicity detection}). Note that a {\em strong} form of such an ability
is that the robot are able to count the exact number of robots located at the same position.
A {\em weaker} form consists in considering the detector as an abstract device able
to say if any robot location contains either exactly one or more than one robot.  

In this paper, we investigate the possibility to {\em deterministically} solve the gathering problem
with {\em weak robots} (\ie anonymous, autonomous, disoriented, deaf and dumb, and oblivious).  
This problem has been extensively studied in the literature assuming various 
settings.  For instance, the robots move either among the nodes of a graph~\cite{FlocchiniKKSS04,KlasingMP08}, 
or in the plane~\cite{AP06,AOSY99,CFPS03,FPSW05,P01,Prencipe07,SY99},   
their visibility can be limited (visibility sensors are supposed to be accurate within a constant range, 
and sense nothing beyond this range)~\cite{FPSW05,%AOSY99,
SDY09}, robots are prone to faults~\cite{AP06,DefagoGMP06}.

%Another key parameter is the degree of synchrony. Each 
%robot operates in Look-Compute-Move cycles where Look operation consists only in 
%observing/sensing the position of the others robots in the plane, and based on sensed observation, 
%the robot computes a destination toward which the robot attend to move. 
%There are mainly two types of synchrony models.  One is called {\em SSM} (or {\em ATOM})---SSM
%stands for Semi-Synchronous Model---\cite{SY99}, the other CORDA~\cite{P01}.  
%In both model, in every system configuration, each robot is either {\em active} or {\em inactive}.
%In SSM, cycles are executed synchronously in rounds by all active robots~\cite{}.
%In CORDA, the cycles are executed asynchronously, \ie the three operations arbitrarily 
%interleave between any pair of robots~\cite{}.  
%%Notice that SSM being stronger than CORDA (any execution in SSM can be emulated in CORDA, not 
%%the reverse), every impossibility result in SSM also holds in CORDA.  By opposition, every 
%%solution in CORDA also works in SSM.cite{}.

In this paper, we address the stabilization aspect of the gathering problem. 
A deterministic system is (self-)stabilizing if, regardless of the initial 
states of the computing units, it is guaranteed to converge to the intended behavior in a finite number of 
steps~\cite{D00}.
To our best knowledge, all the above solutions assume that in the initial configuration, 
no two robots are located at the same position.  So, effectively, as already noticed in~\cite{DK02,DP09}, 
this implies that none of them is ``truly'' self-stabilizing---initial configurations where
robots are located at the same positions are avoided.  Note that surprisingly, such an assumption prevents
to initiate the system where the problem is solved, \ie initially all the robots occupy the same position. 

In this paper, we study the gathering problem assuming any arbitrary initial configurations, that is
in which some robots can share the same positions.  Clearly, assuming weak multiplicity detection
(each robot location contains either exactly one or more than one robot), the problem cannot be solved 
deterministically.  Informally, if all the robots are at exactly two positions, then there is
no way to maintain a particular position as an invariant.  So, there are some executions where the system behaves
as if it contains exactly two robots, leading to the impossibility result in~\cite{SY99}.    
We introduce the concept of strong multiplicity detection---the robot are able to count the exact 
number of robots located at the same position.   Even with such capability, the problem cannot be solved 
deterministically, if the number of robots is even.  The proof is similar as above: If initially the 
robots occupy exactly two positions, then there is no way to maintain a particular position as an invariant.
Again, the impossibility result in~\cite{SY99} holds.
By contrast, we show that with an odd number of robots, the problem is solvable.  Our proof is constructive, as 
we present and prove a deterministic algorithm for that problem.  The proposed solution has the
nice property of being self-stabilizing since no initial configuration is excluded.  

In the next section (Section~\ref{sec:model}), we describe the distributed 
system and the problem we consider in this paper. 
Our main result with its proof is given in Section~\ref{sec:g}. 
We conclude this paper in Section~\ref{sec:conc}.
Due to the lack of space, some proofs have been moved in the Annexes section. 

%\begin{definition}
%\label{def:prob}
%Given $n$ robots $r_1,\cdots,r_n$, arbitrarily placed in the plane, make them gather at one point in a finite time.
%\end{definition}

\section{Preliminaries}
\label{sec:model}

In this section, we define the distributed system and the problem considered in this paper.

\subsection{Distributed Model.}
We adopt the semi-synchoronous model introduced in~\cite{SY96}, below referred to as~$SSM$.  
The \emph{distributed system} considered in this paper consists of $n$ robots  
$r_{1}, r_{2},\cdots , r_{n}$---the subscripts $1,\ldots ,n$ are used for notational purpose only.
Each robot $r_{i}$, viewed as a point in the Euclidean plane, moves on this two-dimensional 
space unbounded and devoid of any landmark.  It is assumed that two or more robots may simultaneously occupy the same physical location.
%When no ambiguity arises, $r_{i}$ also denotes the point in the plane occupied by that robot. 

Any robot can observe, compute and move with infinite decimal precision.
The robots are equipped with sensors enabling to detect the instantaneous position of the other robots in the plane. 
In particular, we distinguish two types of multiplicity detection : \emph{weak multiplicity detection} and  \emph{strong multiplicity detection}.

\begin{definition}[Weak multiplicity detection]\cite{CFPS03,FlocchiniIPS08}
\label{def:weak}
The robots have weak multiplicity detection if, for every point $p$, their sensors can detect if there is no robot, there is one robot, or there are more than one robot. In the latter case, the robot might not be capable of determining the exact number of robots.
\end{definition}

\begin{definition}[Strong multiplicity detection]
\label{def:strong}
The robots have strong multiplicity detection if, for every point $p$, their sensors can detect the number of robots on $p$.
\end{definition}

%In this paper, we will discuss the influence of weak multiplicity detection and strong multiplicity detection.

%Obviouly, if the robots have strong multiplicity detection then they have weak multiplicity detection but the reverse it is not true.  
  
Each robot has its own local coordinate system and unit measure.  
The robots do not agree on the orientation of the axes of their local coordinate system, 
nor on the unit measure. 
They are \emph{uniform} and \emph{anonymous}, i.e, they all have the same program using no 
local parameter (such that an identity) 
allowing to differentiate any of them.  
%They share no kind of common coordinate mechanism nor common sense of direction (e.g., compass).
They communicate only by observing the position of the others and they are \emph{oblivious}, i.e.,
none of them can remember any previous observation nor computation performed 
in any previous step. 

Time is represented as an infinite sequence of time instants $0, 1, \ldots, j, \ldots$ 
Let $\mathcal{P}(t)$ be the set of the positions in the plane occupied by the $n$ 
robots at time $t$. For every $t$, $\mathcal{P}(t)$ is called the \emph{configuration} 
of the distributed system in $t$. Given any point $p$, $|p|$ denotes the number of robots located on $p$. Note that, if the robots do not have the multiplicity detection then $|p|\leq1$ for all the robots.
$\mathcal{P}(t)$ expressed in the local coordinate system of any robot $r_i$ is called a \emph{view}.%, denoted 
%$v_i(t_j)$. 
At each time instant $t$, each robot $r_i$ is either {\it active} or {\it inactive}. 
The former means that, during the computation \emph{step} $(t,t+1)$, using 
a given algorithm, $r_i$ computes in its local coordinate system a position $p_i(t+1)$ depending 
only on the system configuration at $t$, and moves towards $p_i(t+1)$---$p_i(t+1)$ can be equal to
$p_i(t)$, making the location of $r_i$ unchanged.
In the latter case, $r_i$ does not perform any local computation and remains at the same position. 
 In every single activation, the distance 
traveled by any robot $r$ is bounded by $\sigma_r$. So, if the destination point computed by $r$ is farther than $\sigma_r$, then $r$ moves 
toward a point of at most $\sigma_r$. This distance may be different between two robots.

The concurrent activation of robots is modeled by 
the interleaving model in which the robot activations are driven by a \emph{fair scheduler}.  
At each instant $t$, the scheduler 
arbitrarily activates a (non empty) set of robots.  
Fairness means that every robot is infinitely often activated by the scheduler.

\subsection{Specification}
The \emph{ Gathering Problem} ($\mathcal{GP}$) is to design a distributed protocol $P$ for $n$ mobile robots so that the following properties are true : 

\begin{itemize}
\item {\emph{Convergence}:} Regardless of the initial positions of the robots on the plane, all the robots are
located at the same position in finite time. 
\item {\emph{Closure}:} Starting from a configuration where all the robots are located at the same position, all the robots
are located at the same position thereafter. 
\end{itemize}

%\section{Gathering with weak multiplicity detection}

%\begin{theorem}
%With weak multiplicity detection, there exists no deterministic self-stabilizing algorithm solving the gathering problem for $n$ robots, $n$ even or odd.
%\end{theorem}

\section{Gathering with strong multiplicity detection}
\label{sec:g}

In this section, we prove the following theorem :

\begin{theorem}
\label{theo:gathodd}
With strong multiplicity detection, there exists a deterministic self-stabilizing algorithm solving $\mathcal{GP}$ for $n$ robots if, and only if, $n$ is odd.
\end{theorem}

As mentionned in the introduction, even with strong multiplicity detection there do not exist any deterministic algorithm solving $\mathcal{GP}$ for an even number of robots. So, to prove Theorem~\ref{theo:gathodd} we first give a deterministic self-stabilizing algorithm solving $\mathcal{GP}$ for an odd number of robots having the strong multiplicity detection. Then, we prove the correctness of the algorithm.

%\subsection{Impossibility of a Deterministic Algorithm for an even number of robots}

%\begin{lemma}
%\label{lem:gatheven}
%With strong multiplicity detection, there exists no deterministic self-stabilizing algorithm solving the gathering problem for $n$ robots, if $n$ is even.
%\end{lemma}

\subsection{Deterministic Self-stabilizing Algorithm for an odd number of robots.}
In this subsection, we give a deterministic self-stabilizing algorithm solving $\mathcal{GP}$ for an odd number of robots.  We first 
provide particular notations, basic definitions and properties that we use for symplifying 
the design and proofs of the protocol. Next, the protocol is presented.

\subsubsection{Notations, Basic Definitions and Properties.}

Given a configuration $\mathcal{P}$, $Max\mathcal{P}$ indicates the set of all the points $p$ such that $|p|$ is maximal. In other terms, $\forall p_i \in Max\mathcal{P}$ and $\forall p_j \in \mathcal{P}$, we have $|p_i|\geq |p_j|$. $|Max\mathcal{P}|$ will be the cardinality of $Max\mathcal{P}$.

\begin{remark}
Since the robots have the strong multiplicity detection, then they are able to compute $|p|$ for every point $p\in\mathcal{P}$. In particular, all the robots can determine $Max\mathcal{P}(t)$ at each time instant $t$.  
\end{remark}

Given three distinct points $r,r'$ and $c$ in the plane, we say that the two half-lines $[c,r)$ and $[c,r')$ divide the plane into two \emph{sectors} if and only if  
\begin{itemize}

\item either $r,r'$ and $c$ are not colinear,

\item or $r,r'$ and $c$ are colinear and $c$ is between $r$ and $r'$ on the segment $[r,r']$.
\end{itemize} 

If it exists then this pair of sectors is denoted by $\{\underline{rcr'},\overline{rcr'}\}$ and we assume that the two half-lines $[c,r)$ and $[c,r')$ do not belong to any sector in $\{\underline{rcr'},\overline{rcr'}\}$ . Note that, if the three points $r,r'$ and $c$ are not colinear then one of two sectors is \emph{convex} (angle centered at $c$ between $r$ and $r' \leq 180^o$)  and the other one is \emph{concave} (angle centered at $c$ between $r$ and $r' > 180^o$). Otherwise, the three points $r,r'$ and $c$ are colinear and the two sectors are convex and more precisely they are \emph{straight} (both conjugate angles centered at $c$ between $r$ and $r'$ are equal to $180^o$).

%Let $\mathcal{S}$ and $C$ be respectively a sector from $\{\underline{rcr'},\overline{rcr'}\}$  and a circle centered at $c$. We denote by $arc(\mathcal{S})$ the arc of the circle $C$ between $r$ and $r'$ inside $\mathcal{S}$. 

\begin{definition}[Smallest enclosing circle]\cite{DK02} 
\label{def:sec}
Given a set $\mathcal{P}$ of $n\geq2$ points $p_1,p_2,\cdots,p_n$ on the plane, the smallest enclosing circle of $\mathcal{P}$ , called $SEC(\mathcal{P})$, is the smallest circle enclosing all the positions in $\mathcal{P}$. It passes either through two of the positions that are on the same diameter (opposite positions), or through at least 
three of the positions in $\mathcal{P}$. 
\end{definition}

When no ambiguity arises, $SEC(\mathcal{P})$ will be shortly denoted by $SEC$ and $SEC(\mathcal{P})\cap\mathcal{P}$ will indicate the set of all the points both on $SEC(\mathcal{P})$ and $\mathcal{P}$. Besides, we will say that a robot $r$ is inside $SEC$ if, and only if, there is not located on the circumference of $SEC$.  
In any configuration $\mathcal{P}$, $SEC$ is unique and can be computed in linear time~\cite{Chrystal85}.

Given a set $\mathcal{P}$ of $n\geq2$ points $p_1,p_2,\cdots,p_n$ on the plane and $SEC(\mathcal{P})$ its smallest enclosing circle,  $\mathcal{R}ad(SEC(\mathcal{P}))$ will indicate the length of the radius of $SEC(\mathcal{P})$.

The next lemma contains a simple fact.
\begin{lemma}
\label{lem:triv}
Let $\mathcal{P}_1$ be an arbitrary configuration of $n$ points. Let $\mathcal{P}_2$ be a configuration obtained by pushing inside $SEC(\mathcal{P}_1)$ all the points which are in $\mathcal{P}_1\cap SEC(\mathcal{P}_1)$. We have $\mathcal{R}ad(SEC(\mathcal{P}_2))< \mathcal{R}ad(SEC(\mathcal{P}_1))$.
\end{lemma}

%\begin{proof}
%Assume by contradiction that $SEC(\mathcal{P}(t)$ remains the smallest enclosing circle of  $\mathcal{P}(t_i)$ for all $t_i\geq t$. By fairness, we deduce that there exists  $t_k>t$ such that there are no robot on the circumference of $SEC(\mathcal{P}(t)$. However, by Definition~\ref{def:sec} there exist at least two robots on $SEC$. So, $SEC(\mathcal{P}(t)\ne SEC(\mathcal{P}(t_k)$. Contradiction

%So, $SEC$ is not invariant. Furthermore, since the active robots are allowed to move only inside $SEC$, $SEC(\mathcal{P}(t)$ remains an enclosing circle for all the configurations $\mathcal{P}(t_i)$, $t_i\geq t$ even if it is not the smallest one. Consequently, there exists $t_k>t$ such that $\mathcal{R}ad(SEC(\mathcal{P}(t_k))< \mathcal{R}ad(SEC(\mathcal{P}(t))$
%\end{proof}

Let $\mathcal{S}$ and $C$ be respectively a sector in $\{\underline{pcp'},\overline{pcp'}\}$  and a circle centered at $c$. We denote by $arc(C,\mathcal{S})$ the arc of the circle $C$ inside $\mathcal{S}$. 
Given a set $\mathcal{P}$ of $n\geq2$ points $p_1,p_2,\cdots,p_n$ on the plane and $SEC(\mathcal{P})$ its smallest enclosing circle centered at $c$, we say that $p$ and $p'$ are \emph{adjacent on $SEC(\mathcal{P})$} if, and only if, $p$ and $p'$ are in $\mathcal{P}$ and there exists one sector $\mathcal{S}\in \{\underline{pcp'},\overline{pcp'}\}$  such that there is no point in $arc(SEC(\mathcal{P}),\mathcal{S})\cap\mathcal{P}$.

The following property is fundamental about smallest enclosing circles 

\begin{property}\cite{Cieliebak04}
\label{ciele}
Let $\mathcal{P}$ and $c$ be respectively a set of $n\geq2$ points $p_1,p_2,\cdots,p_n$ on the plane and the center of $SEC(\mathcal{P})$. If $p$ and $p'$ are \emph{adjacent on $SEC(\mathcal{P})$} then, there does not exist a concave sector $\mathcal{S}$ in $\{\underline{pcp'},\overline{pcp'}\}$  such that there is no point in  $arc(SEC(\mathcal{P}),\mathcal{S})\cap\mathcal{P}$.
\end{property}

Property~\ref{lem:sec} is more general than Property~\ref{ciele}

\begin{property}
\label{lem:sec}
Let $\mathcal{P}$ and $c$ be respectively a set of $n\geq2$ points $p_1,p_2,\cdots,p_n$ on the plane and the center of $SEC(\mathcal{P})$. If $p$ and $p'$ are in $\mathcal{P}$ then, there does not exist a concave sector $\mathcal{S}$ in $\{\underline{pcp'},\overline{pcp'}\}$  such that there is no point in $\mathcal{S}\cap\mathcal{P}$.
\end{property}

\begin{proof}
Assume by contradiction that $p$ and $p'$ are in $\mathcal{P}$ and, there exists a concave sector $\mathcal{S}$ in $\{\underline{pcp'},\overline{pcp'}\}$  such that there is no point in $\mathcal{S}\cap\mathcal{P}$. So, there is no point in $arc(SEC(\mathcal{P}),\mathcal{S})\cap\mathcal{P}$. We deduce that there exists a concave sector $\mathcal{S'}$ in $\{\underline{qcq'},\overline{qcq'}\}$  such that $q$ and $q'$ are adjacent on $SEC(\mathcal{P})$ and there is no point in  $arc(SEC(\mathcal{P}),\mathcal{S'})\cap\mathcal{P}$. Contradiction with Property~\ref{ciele}.
\end{proof}

Figure~\ref{fig:lem} illustrates Property~\ref{lem:sec}.

\begin{figure}[!htbp]
\begin{center}
    \epsfig{file=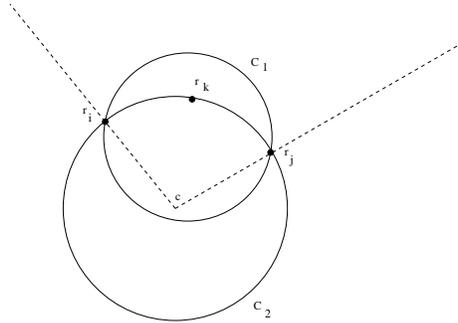, width=0.5\linewidth}
 \caption{$C_2$ is an enclosing circle for the three points $r_i,r_j$ and $r_k$. However, there is no point in the intersection between $C_2$ and the concave sector formed by $r_i$, $r_j$ and the center $c$ of $C_2$. So, $C_2$ can be replace by a smaller enclosing circle, here $C_1$, even if all the points are on the circumference of $C_2$.\label{fig:lem}}
\end{center}
\end{figure}

\begin{observation}
\label{toto}
Given three colinear points, c,r,r'. If $c$ is on the segment $[r,r']$, then $c$ cannot be on the circumference of a circle enclosing $r$ and $r'$.
\end{observation}

\begin{definition}[Convex Hull]\cite{PreparataH77}
\label{def:conv}
Given a set $\mathcal{P}$ of $n\geq2$ points $p_1,p_2,\cdots,p_n$ on the plane, 
the convex hull of $\mathcal{P}$, denoted $H(\mathcal{P})$ , is the smallest polygon such that 
every point in $\mathcal{P}$ is either on an edge of $H(\mathcal{P})$ or inside it.
\end{definition}

Informally, it is the shape of a rubber-band stretched around $p_1,p_2,\cdots,p_n$. 
The convex hull is unique and can be computed with time complexity $O(n \log n)$~\cite{PreparataH77}.
When no ambiguity arises, $H(\mathcal{P})$ will be shortly denoted by $H$ and $H(\mathcal{P})\cap\mathcal{P}$ will indicate the set of the positions both on $H(\mathcal{P})$ and $\mathcal{P}$. 

From Definition~\ref{def:conv}, we deduce the following property :

\begin{property}
\label{obs:hull}
Let $\mathcal{P}$ be respectively a set of $n\geq2$ points that are not on the same line and let $H(\mathcal{P})$ be a convex hull. The two following properties are equivalent  
\begin{enumerate}
\item Any point $c$, not necessarily in $\mathcal{P}$, is located on $H$ (either on a vertice or an edge) 

\item there is a concave or a straight sector $\mathcal{S}$ in $\{\underline{rcr'},\overline{rcr'}\}$  such that $r$ and $r'$ are in $\mathcal{P}$ and there exists no point $\in \mathcal{P}\cap\mathcal{S}$.
\end{enumerate}
\end{property}

The relationship between the smallest enclosing circle and the convex hull is given by the following property  

\begin{property}\cite{Chrystal85}
\label{lem:sechull}
Given a set $\mathcal{P}$ of $n\geq2$ points on the plane. We have $$SEC(\mathcal{P})\cap\mathcal{P}\subseteq H(\mathcal{P})\cap\mathcal{P}$$.
\end{property}

\subsubsection{The Algorithm}

Based on the definitions and basic properties introduced above, we are now ready to present a deterministic self-stabilizing algorithm that allows $n$ robots ($n$ odd) to gather in a point, regardless of the initial positions of the robots on the plane. The idea of our algorithm is as follows : It consists in transforming an arbitrary configuration $\mathcal{P}$ into one where there is exactly one point $p_{max}\in Max\mathcal{P}$. When such a configuration is reached, all the robots which are not located at $p_{max}$ move towards $p_{max}$ avoiding to create another point $q$ than $p_{max}$ such that $|q|\geq p_{max}$. 

When $|Max\mathcal{P}|\ne 1$, we will distinguish two cases : $|Max\mathcal{P}|= 2$ and $|Max\mathcal{P}|\geq 3$.
  
If $Max\mathcal{P}=\{p_{max1};p_{max2}\}$, then each robot which is not located neither on $p_{max1}$ nor $p_{max2}$ moves towards its closest position $\in Max\mathcal{P}$ by avoiding to create an adding maximal point. Since the number of robots is odd, we have eventually either $|p_{max1}|>|p_{max2}|$ or $|p_{max1}|>|p_{max2}|$ and then, $|Max\mathcal{P}|=1$. 

For the case $|Max\mathcal{P}|\geq 3$, our strategy consists in trying to create a unique maximal point inside $SEC$. To reach such a configuration, we distinguish three subcases :
\begin{enumerate}
\item If there is no robot inside $SEC$, then all the robots are allowed to move towards the center of $SEC$.
\item If all the robots inside $SEC$ are located at the center of $SEC$, then only the robots located in $SEC\cap Max\mathcal{P}$ are allowed to move towards the center of $SEC$.
\item If some robots inside $SEC$ are not located at the center of $SEC$, then only the robots inside $SEC$ are allowed to move towards the center of $SEC$.
\end{enumerate}

The main algorithm is shown in Algorithm~\ref{algo:main}. In Algorithm~\ref{algo:main}, we use two subroutines :  $move\_to\_carefully(p)$ and $choose\_closest\_position(p_1,p_2)$. The former allows a robot $r$, located at $q$, to move towards $p$ only if there is no robot on the segment $[q,p]$ except the robots located on $p$ or the robots located on $q$. The latter one returns the closest position to $r$ among $\{p_1,p_2\}$. If the distance between $r$ and $p_1$ is equal to the distance between $r$ and $p_2$ then the function returns $p_1$. 

\begin{algo}[!htb]
%\begin{footnotesize}
%\begin{small}
\begin{tabbing}
xxxxx \= xxxxx \= xxxxx \= xxxxx \= xxxxx \= xxxxx \= xxxxx \= xxxxx \= xxxxx \= \kill 
\NA \kill \\
\> $\mathcal{P}:=$ the set of all the positions; \\
\> $Max\mathcal{P}:=$ the set of all the points $p\in\mathcal{P}$ such that $|p|$ is maximal; \\
\> \IF{$|Max\mathcal{P}|=1$} \\ 
\> \THEN  \> $p_{max}:=$ the unique point in $Max\mathcal{P}$;\\
\> \>  \IF{ I am not on $p_{max}$}; \\
\> \> \THEN $move\_to\_carefully(p_{max})$; \\
\> \> \ENDIF \\
\> \ENDIF \\
\> \IF{$|Max\mathcal{P}|=2$} \\
\>  \THEN  \> $p_{max1}:=$ the first point in $Max\mathcal{P}$; \\
\> \>   $p_{max2}:=$ the second point in $Max\mathcal{P}$; \\
\> \>  \IF{ I am not neither on $p_{max1}$ nor $p_{max2}$}  \\
\> \> \THEN $q:= choose\_closest\_position(p_{max1},p_{max2})$; \\
\> \> \> $move\_to\_carefully(q)$; \\
\> \> \ENDIF \\
\> \ENDIF \\
\> \IF{$|Max\mathcal{P}|\geq3$} \\
\> \THEN $SEC:=$ the smallest circle enclosing all the points in $\mathcal{P}$;\\
\> \> $c:=$ the center of $SEC$;\\
\> \> $Boundary:= SEC\cap\mathcal{P}$;\\
\> \> $Inside:= \mathcal{P} \setminus Boundary$;\\
\> \> \IF{ $Inside \ne \emptyset$ } \\
\> \> \THEN \IF { All the robots $\in Inside$ are located at $c$} \\
\> \> \> \THEN  \IF {I am in $(Boundary\cap Max\mathcal{P})$} \\
\> \> \> \> \THEN $move\_to(c)$; \\
\> \> \> \> \ENDIF \\
\> \> \> \ELSE \IF{ I am in $Inside$}\\
\> \> \> \> \THEN $move\_to(c)$; \\
\> \> \> \> \ENDIF \\
\> \> \> \ENDIF \\
\> \> \ELSE $move\_to(c)$; \\
\> \> \ENDIF \\
\> \ENDIF
\end{tabbing}
%\end{small}
%\end{footnotesize}
\caption{Gathering for an odd number of robots, executed by each robot. \label{algo:main}}
\end{algo} 

\subsubsection{Proof of closure}
\begin{lemma}[Closure]
\label{lem:clos}
According to Algorithm~\ref{algo:main}, if all the robots are located at the same position $p$, then all the robots are located at the same position thereafter. 
\end{lemma}

\subsubsection{Proof of convergence}

%From Lemma~\ref{lem:gatheven}, it remains to give a deterministic self-stabilizing algorithm solving the gathering problem for $n$ robots when $n$ is odd. Clearly, starting from a configuration  $Max\mathcal{P}$ wherein $|Max\mathcal{P}|=1$ with $p_{max}\in Max\mathcal{P}$, it suffices that all the robots that are not located at $p_{max}$ move toward $p_{max}$ by avoiding creating another position $p$ such that $|p|\geq|p_{max}|$. In order to achieve that, at each time instant, it is enough to require that any robot $r$ is allowed to move in straight line toward $p_{max}$ only if there exists no robot on the segment $[r,p_{max}]$ except, of course, $r$ and the robots located at $p_{max}$. So it remains to describe procedure  leading a group of robots into a configuration where $|Max\mathcal{P}|=1$
%Let $p_{max_{1}}$ and $p_{max_{2}}$ be the both positions in $Max\mathcal{P}$. All the robots that are not located at $p_{max_{1}}$ or $p_{max_{2}}$ move toward the closest position in $Max\mathcal{P}$ by avoiding creating another position $p$ such that $|p|\geq|p_{max}|$. In the case where any robot $r$ is at the same distance from $p_{max_{1}}$ and $p_{max_{2}}$ then $r$ arbitrarily chooses to move either toward $p_{max_{1}}$ or toward $p_{max_{2}}$. Clearly, $|Max\mathcal{P}|$ can remain equal to $2$ but only a finite time because the number of robots is odd. So, $|Max\mathcal{P}|=1$ eventually.
\paragraph{Cases $|Max\mathcal{P}|=1$ and $|Max\mathcal{P}|=2$.}  

\begin{lemma}
\label{lem:cas1}
Let $\mathcal{P}$ be an arbitrary configuration for an odd number of $n$ robots. According to Algorithm~\ref{algo:main}, if $|Max\mathcal{P}|=1$ then all the robots are located at the same position in finite time.
\end{lemma}

\begin{lemma}
\label{lem:cas2}
Let $\mathcal{P}$ be an arbitrary configuration for an odd number of $n$ robots. According to Algorithm~\ref{algo:main}, if $|Max\mathcal{P}|=2$ then all the robots are located at the same position in finite time.
\end{lemma}

\paragraph{Case $|Max\mathcal{P}|\geq3$.}
In this paragraph, we prove that starting from a configuration where $|Max\mathcal{P}|\geq3$, all the robots are located at the same position in finite time. More precisely, we consider the case where there exists at least one robot inside $SEC(\mathcal{P}(t))$ ( refer to Lemma~\ref{tyty}) and the case where there is no robot inside $SEC(\mathcal{P}(t))$ ( refer to Lemma~\ref{lem:der}). 

In order to prove Lemma~\ref{tyty}, we use Lemmas~\ref{lem:tam} and~\ref{truc}. In particular, Lemmas~\ref{lem:tam} shows that, under specific conditions, the center of $SEC(\mathcal{P}(t))$ is inside $SEC(\mathcal{P}(t+1))$ even if $SEC(\mathcal{P}(t))\ne SEC(\mathcal{P}(t+1))$ or the center of $SEC(\mathcal{P}(t))$ is not the center of $SEC(\mathcal{P}(t+1))$. The proof of Lemma~\ref{lem:tam} is organized in two parts. In the former one, we consider the case where the center of $SEC$ is also on the convex hull (see Figure~\ref{coS}.c). In the latter one, we consider the case where the center of $SEC$ is not on the convex hull.

\begin{figure}[!htbp]
\begin{center}
  \begin{minipage}[t]{0.4\linewidth}
    \centering
    \epsfig{file=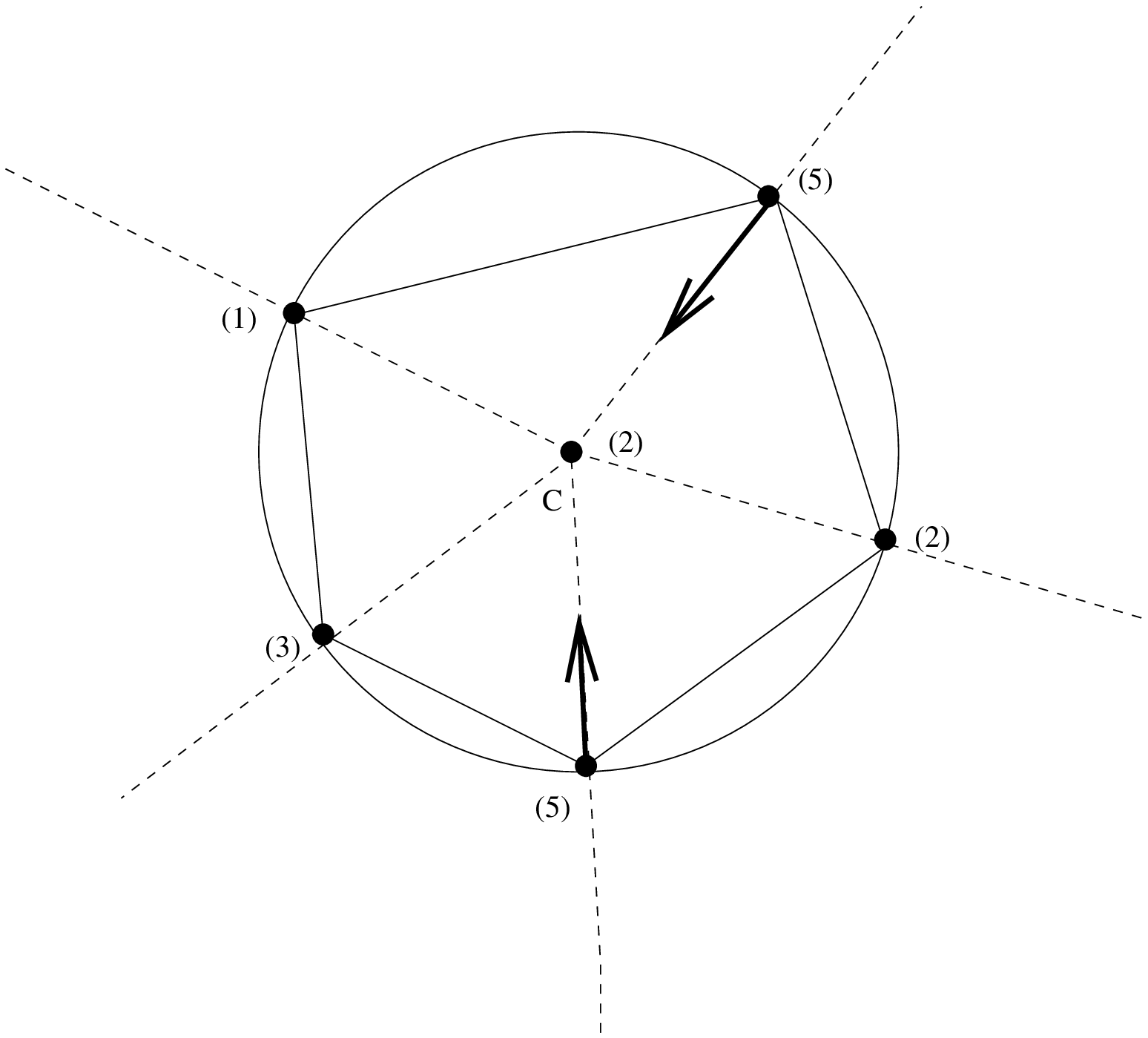, width=1.2\linewidth}\\
    {\footnotesize ($a$)}
  \end{minipage}
  \begin{minipage}[t]{0.4\linewidth}
    \centering
    \epsfig{file=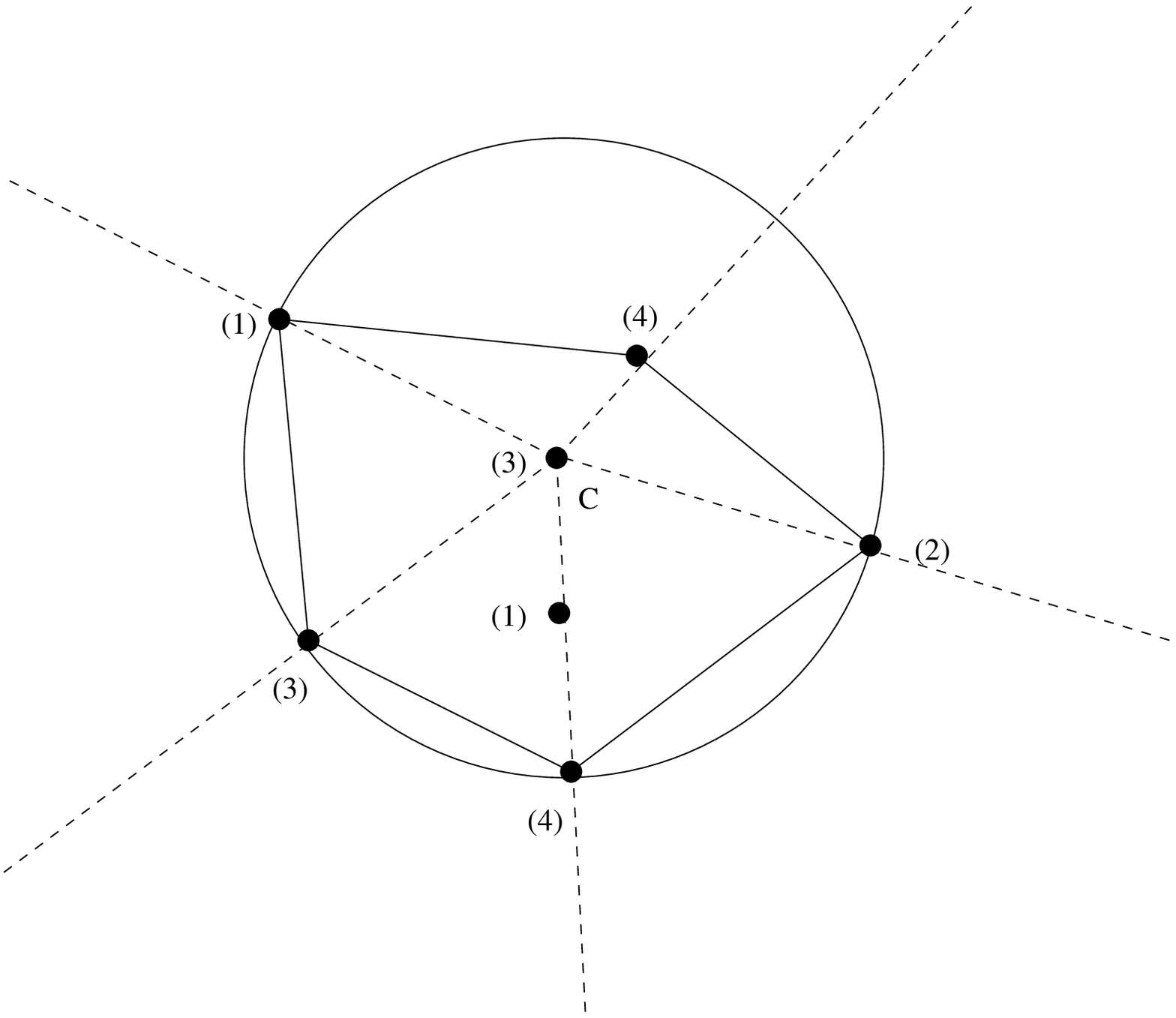, width=1.2\linewidth}\\
    {\footnotesize ($b$)}
  \end{minipage}
%  \begin{minipage}[t]{0.3\linewidth}
%    \centering
%    \epsfig{file=convexSec.eps, width=0.7\linewidth}\\
%    {\footnotesize ($c$)}
%  \end{minipage}
\end{center}
 \caption{The numbers between parenthesis indicate the multiplicity. In Figure~$a$, we have a configuration $\mathcal{P}(t)$ where the center $c$ of $SEC(\mathcal{P}(t))$ is inside the convex hull. Figure~$b$, we have configuration $\mathcal{P}(t+1)$ where some robots have moved toward $c$ and $c$ is inside the new convex hull.\label{coS}}
\end{figure}

\begin{lemma}
\label{lem:tam}
Let $\mathcal{P}(t)$ be a configuration such that $|Max\mathcal{P}|\geq3$ and there exists at least one robot inside $SEC(\mathcal{P}(t))$.

According to Algorithm~\ref{algo:main}, if both conditions are true :

\begin{enumerate}
\item some robots $\in\mathcal{P}(t)\cap SEC(\mathcal{P}(t))$ move in straight line toward the center $c$ of $SEC(t)$ and,  
\item for every $p$ $\in\mathcal{P}(t)\cap SEC(\mathcal{P}(t))$ there exists at least one robot in $p$ which does not reach $c$ at time $t+1$  
\end{enumerate}
then, the center of $SEC(\mathcal{P}(t))$ is inside $SEC(\mathcal{P}(t+1))$ at time $t+1$.
\end{lemma}

\begin{proof}
Let $c$ be the center of $SEC(t)$ at time $t$. We consider two cases, depending on whether $c$ is on the convex hull $H(\mathcal{P}(t))$ or not, at time $t$.
\begin{itemize}

\item {\bf $c$ is on $H(\mathcal{P}(t))$ at time $t$.} From Property~\ref{obs:hull}, there exists a concave or a straight sector $\mathcal{S}$ in $\{\underline{xcy},\overline{xcy}\}$  such that $x$ and $y$ are in $\mathcal{P}(t)$ and there is no point $\in \mathcal{P}(t)\cap\mathcal{S}$. However, from Property~\ref{lem:sec}, we know that there do not exist two points $x$ and $y$ in $\mathcal{P}(t)$ such that there exists a concave sector $\mathcal{S}$ in $\{\underline{xcy},\overline{xcy}\}$  and $\mathcal{P}(t)\cap\mathcal{S}=\emptyset$. So, 
there exists {\bf only} a straight sector $\mathcal{S}$ in $\{\underline{xcy},\overline{xcy}\}$  such that $x$ and $y$ are in $\mathcal{P}(t)$ and there is no point $\in \mathcal{P}(t)\cap\mathcal{S}$. Consequently, $c$ is on the segment $[x,y]$ at time $t$. Since the robots move in straight line towards $c$ and since there exist some robots located at $x$ and some robot located at $y$ which do not reach $c$ at time $t+1$ then, $c$ is on the segment $[r,s]$ at time $t+1$ with $r$ and $s$ $\in\mathcal{P}(t+1)$.
From Observation~\ref{toto}, we deduce that $c$ is inside $SEC(\mathcal{P}(t+1))$ at time $t+1$.

\item {\bf $c$ is not on $H(\mathcal{P}(t))$ at time $t$.} In this case, all the points in $\mathcal{P}(t)$ are not on the same line otherwise $c$ would have been on  $H(\mathcal{P}(t))$.  So, from Property~\ref{obs:hull} we know that there does not exist a concave or a straight sector $\mathcal{S}$ in $\{\underline{xcy},\overline{xcy}\}$  such that $x$ and $y$ are in $\mathcal{P}(t)$ and there is no point $\in \mathcal{P}(t)\cap\mathcal{S}$. Since the robots move in straight line towards $c$ and since for each point $p\in \mathcal{P}(t)$ there exists at least one robot located on $p$ which does not reach $c$ at time $t+1$ then, we deduce that there does not exist a concave or a straight sector $\mathcal{S}$ in $\{\underline{rcs},\overline{rcs}\}$  such that $r$ and $s$ are in $\mathcal{P}(t+1)$ and there is no point $\in \mathcal{P}(t+1)\cap\mathcal{S}$ (Figures~\ref{coS}.$a$ and~\ref{coS}.$b$ illustrate this fact). So, from Property~\ref{obs:hull} $c$ is inside $H(\mathcal{P}(t+1))$ at time $t+1$, and from Lemma~\ref{lem:sechull} we deduce that $c$ is inside $SEC(\mathcal{P}(t+1))$.
\end{itemize}

\end{proof}

\begin{lemma}
\label{truc}
Let $\mathcal{P}(t)$ be a configuration such that $|Max\mathcal{P}|\geq3$ and there exists at least one robot inside $SEC(\mathcal{P}(t))$.
If any robot $r$ is inside $SEC(\mathcal{P}(t))$ and $r$ is located on the boundary of $SEC(\mathcal{P}(t+1))$ then $|Max\mathcal{P}(t+1)|\leq2$.
\end{lemma}

\begin{proof}
By contradiction assume that $r$ is inside $SEC(\mathcal{P}(t))$ and $r$ is located on the boundary of $SEC(\mathcal{P}(t+1))$ and $|Max\mathcal{P}(t+1)|>2$. Let $c$ be the center of $SEC(\mathcal{P}(t))$ at time $t$. From assumption, some robots on the boundary of $SEC(\mathcal{P}(t))$ have moved toward the center of $SEC(\mathcal{P}(t))$. According to Algorithm~\ref{algo:main}, that implies that all the robots inside $SEC(\mathcal{P}(t))$, notably $r$, are located at the center of $SEC(\mathcal{P}(t))$ at time $t$. So, $c$ is on the boundary of $SEC(\mathcal{P}(t+1))$. From Lemma~\ref{lem:tam}, we deduce that  there exists a point $p$ $\in\mathcal{P}(t)\cap SEC(\mathcal{P}(t))$ such that all the robots in $p$ have reached $c$ at time $t+1$. However, according to Algorithm~\ref{algo:main} only the robots located in $\in Max\mathcal{P}(t)\cap SEC(\mathcal{P}(t))$ are allowed to move at time $t$. Therefore, for every point $p\ne c$ we have $|c|>|p|$ at time $t+1$. Hence, $|Max\mathcal{P}(t+1)|=\{c\}$ i.e., $|Max\mathcal{P}(t+1)|=1$. A contradiction.
\end{proof}

\begin{lemma}
\label{tyty}
Let $\mathcal{P}(t)$ a configuration such that $|Max\mathcal{P}|\geq3$ and there exists at least one robot inside $SEC(\mathcal{P}(t))$.
According to Algorithm~\ref{algo:main}, all the robots are located at the same position in finite time.
\end{lemma}

\begin{proof}
Assume by contradiction $|Max\mathcal{P}|\geq3$ forever. From Lemma~\ref{truc}, we know that the robots inside $SEC(\mathcal{P}(t))$ are inside $SEC(\mathcal{P}(t+1))$. So, by induction we deduce that the robots inside $SEC(\mathcal{P}(t))$ are inside $SEC(\mathcal{P}(t_i))$ for all $t_i\geq t$. From Lemma~\ref{lem:triv}, fairness and because of the fact that each robot $r$ can move to at least a constant distance $\sigma_r>0$ in one step, we know that there exists a time instant $t_k$ where the number of robots at the center of $SEC(\mathcal{P}(t_k))$ will be greater than the number of robots not located at the center of $SEC(\mathcal{P}(t_k))$. So, $|Max\mathcal{P}|=1$ : contradiction. So, $|Max\mathcal{P}|\leq2$ in finite time and from Lemmas~\ref{lem:cas1} and~\ref{lem:cas2} all the robots are located at the same position in finite time.
\end{proof}

\begin{lemma}
\label{lem:der}
Let $\mathcal{P}(t)$ be a configuration such that $|Max\mathcal{P}|\geq3$ and there exists no robot inside $SEC(\mathcal{P}(t))$.
According to Algorithm~\ref{algo:main}, all the robots are located at the same position in finite time.
\end{lemma} 

\begin{proof}
 According to Algorithm~\ref{algo:main}, all the robots may decide to move toward the center of $SEC$. Since each robot $r$ can move to at least a constant distance $\sigma_r>0$ in one step, if all the robots are always on the boundary of $SEC(\mathcal{P})$ then, by fairness, the gathering problem is solved in finite time.
Otherwise, 
\begin{itemize}
\item either there exists $t_k>t$ such that $|Max\mathcal{P}(t_k)|\geq3$ and there exists at least one robot inside $SEC(\mathcal{P}(t))$ : From Lemma~\ref{tyty}, we deduce that all the robots are located at the same position in finite time,

\item or there exists $t_k>t$ such that $|Max\mathcal{P}(t_k)|\leq2$ :  from Lemmas~\ref{lem:cas1} and~\ref{lem:cas2} all the robots are located at the same position in finite time.
\end{itemize}
\end{proof} 

\section{Conclusion}
\label{sec:conc}

Assuming strong multiplicity detection, we provided a complete characterization 
(necessary and sufficient conditions) to solve the gathering problem. 
Note that we do not know whether strong multiplicity detection is a necessary condition
to solve the gathering problem. 
In future works, we would like to study the weakest minimal multiplicity detection 
that solves this problem and under which conditions. 
Note that the gathering problem seems to be the only positioning problem that
can be deterministically and self-stabilizing solved.  Indeed, since initially the
robots can share the same positions, there exists no deterministic algorithm to
scatter them in the plane~\cite{DP09}.

\begin{footnotesize}
\begin{small}
\bibliographystyle{plain} %{alpha}
\bibliography{ngon}
\end{small}
\end{footnotesize}
\newpage
\section*{Annexes}

\paragraph{\bf Lemma~\ref{lem:clos}.}
According to Algorithm~\ref{algo:main}, if all the robots are located at the same position $p$, all the robots are located at the same position thereafter.

\paragraph{\bf Proof.}
If all the robots are located at the same position $p$, then $|Max\mathcal{P}|=1$ and all the robots are located at the unique position $p\in Max\mathcal{P}$. According to Algorithm~\ref{algo:main}, in the case $|Max\mathcal{P}|=1$ the robots located on $p$ remains idle. So, all the robots are located at the same position forever.

\paragraph{\bf Lemma~\ref{lem:cas1}.} Let $\mathcal{P}$ be an arbitrary configuration for an odd number of $n$ robots. According to Algorithm~\ref{algo:main}, if $|Max\mathcal{P}|=1$ then all the robots are located at the same position in finite time.
\paragraph{\bf Proof.}
Let $p_{max}$ be the unique point in $Max\mathcal{P}(t)$. According to Algorithm~\ref{algo:main}, the robots located on $p_{max}$ during \emph{step} $(t,t+1)$ remains idle. Moreover, according to Algorithm~\ref{algo:main} and Function $move\_to\_carefully()$, if two robots $r_i$ and $r_j$ are not at the same point at time $t$, i.e., $p_i(t)\ne p_j(t)$ then $p_i(t+1)\ne p_j(t+1)$ at time $t+1$ unless they have reached $p_{max}$. Hence, $p_{max}$ remains the unique point in $Max\mathcal{P}(t_k)$, for all $t_k\geq t$. So, according to Algorithm~\ref{algo:main} and by fairness, we deduce that $|p_{max}|=n$ in finite time.

\paragraph{\bf Lemma~\ref{lem:cas2}.}
Let $\mathcal{P}$ be an arbitrary configuration for an odd number of $n$ robots. According to Algorithm~\ref{algo:main}, if $|Max\mathcal{P}|=2$ then all the robots are located at the same position in finite time.

\paragraph{\bf Proof.}
The proof is organized as follows : 
First, we  prove that there exists  $t_k\geq t$ such that $|Max\mathcal{P}(t_k)|\ne 2$. Then, we prove that there does not exist  any time $t_k\geq t$ such that $|Max\mathcal{P}(t_k)|\geq 3$. Finally, we deduce that Lemma~\ref{lem:cas1} holds.
\begin{enumerate}
\item Assume by contradiction that there does not exist any time $t_k\geq t$ such that $|Max\mathcal{P}(t_k)|\ne 2$. Consequently, for every $t_k \geq t$, $|Max\mathcal{P}(t_k)|= 2$. Let $p_{max1}$ and $p_{max2}$ be the two points in $Max\mathcal{P}(t)$ at time $t$. According to Algorithm~\ref{algo:main}, the robots located either on $p_{max1}$ or on $p_{max2}$ during \emph{step} $(t,t+1)$ remains idle.
Moreover, according to Algorithm~\ref{algo:main} and Function $move\_to\_carefully()$, if two robots $r_i$ and $r_j$ are not at the same point at time $t$, i.e., $p_i(t)\ne p_j(t)$ then $p_i(t+1)\ne p_j(t+1)$ at time $t+1$ unless either $r_i$ and $r_j$ have reached $p_{max1}$ or $r_i$ and $r_j$ have reached $p_{max2}$. So, by induction we deduce that $p_{max1}$ and $p_{max2}$ remains the only positions in $Max\mathcal{P}(t_k)$ for every $t_k\geq t$. By fairness, we deduce that, all the robots are either at $p_{max1}$ or at $p_{max2}$ in finite time. However, since the number of robots is odd then, we have either $|p_{max1}|>|p_{max2}|$ or $|p_{max1}|<|p_{max2}|$. Hence, $|Max\mathcal{P}(t_k)|=1$ : Contradiction.
\item Assume by contradiction that there exists $t_k\geq t$ such that $|Max\mathcal{P}(t_k)|\geq 3$. Without lost of generality, we assume that $t_k$ is the first time for which $|Max\mathcal{P}(t_k)|\geq 3$. Clearly, there does not exist any time $t_l$ such that $t<t_l<t_k$ and $|Max\mathcal{P}(t_k)|=1$ : Indeed from Lemma~\ref{lem:clos} and the proof of Lemma~\ref{lem:cas1}, once there exist a unique point $p_{max}$ then, it remains the unique point in $Max\mathcal{P}$ forever and that would be a contradiction.

Hence, $|Max\mathcal{P}(t_k-1)|= 2$. 

Let $p_{max1}$ and $p_{max2}$ be the two points in $Max\mathcal{P}(t_k-1)$ at time $t_k-1$. According to Algorithm~\ref{algo:main}, the robots located either on $p_{max1}$ or on $p_{max2}$ during \emph{step} $(t,t+1)$ remains idle. Besides,  according to Algorithm~\ref{algo:main} and Function $move\_to\_carefully()$, if two robots $r_i$ and $r_j$ are not at the same point at time $t_k-1$, i.e., $p_i(t_k-1)\ne p_j(t_k-1)$ then $p_i(k)\ne p_j(k)$ at time $t_k$ unless either $r_i$ and $r_j$ have reached $p_{max1}$ or $r_i$ and $r_j$ have reached $p_{max2}$. So, $|Max\mathcal{P}(t_k)|\leq2$ at time $t_k$. A contradiction.
\end{enumerate}

From above, we deduce that if $|Max\mathcal{P}(t)|=2$ at time $t$ then, according to Algorithm~\ref{algo:main} there exists $t_k$, $t_k>t$ such that $|Max\mathcal{P}(t_k)|=1$. So, from Lemma~\ref{lem:cas1}, we know that all the robots will be located at the same position in finite time. 
\end{document}